\pdfoutput=1
\documentclass[runningheads]{llncs}
\usepackage{amsmath}
\usepackage{mathtools}
\usepackage[hidelinks]{hyperref}
\usepackage[capitalise, noabbrev]{cleveref}
\usepackage{microtype}
\usepackage{xspace}
\usepackage[table]{xcolor}
\usepackage[autostyle=false, style=english]{csquotes}
\MakeOuterQuote{"}
\DeclareMathOperator{\MOLS}{MOLS}
\DeclareMathOperator{\N}{N}
\newcommand{\mols}[2]{#1 \MOLS(#2)}
\newcommand{\etal}{\textit{et al.}\xspace}

\newcommand{\df}[1]{\hfill{#1}\hfill}

\author{Noah Rubin\inst{1} \and
Curtis Bright\inst{2} \and
Kevin K. H. Cheung\inst{1} \and
Brett Stevens\inst{1}}
\title{Integer and Constraint Programming Revisited for Mutually Orthogonal Latin Squares}
\titlerunning{IP and CP Revisited for MOLS}
\authorrunning{N.~Rubin, C.~Bright, K.~Cheung, B.~Stevens}
\institute{Carleton University, School of Mathematics and Statistics \and
University of Windsor, School of Computer Science
}

\begin{document}
\maketitle
\setcounter{footnote}{0}
\begin{abstract}
In this paper we provide results on
using integer programming (IP) and constraint programming (CP)
to search for sets of mutually orthogonal latin squares (MOLS).
Both programming paradigms have previously successfully been used to search for
MOLS, but solvers for IP and CP solvers have significantly improved in recent
years and data on how modern
IP and CP solvers perform on the MOLS problem is lacking.  Using
state-of-the-art solvers as black boxes we were able to
quickly find pairs of MOLS (or prove their nonexistence)
in all orders up to ten.  Moreover, we improve the effectiveness
of the solvers by formulating an extended symmetry breaking
method as well as an improvement to the straightforward CP encoding.
We also analyze the effectiveness of using CP and IP solvers
to search for
triples of MOLS, compare our timings to those which have been
previously published, and
estimate the running time of using this approach to resolve the
longstanding open problem of determining the existence of a triple
of MOLS of order ten.
\end{abstract}

\section{Introduction and Motivation}
A \emph{latin square of order $n$} is an $n\times n$ array, $L$, of symbols $\{0,1,\dotsc,n-1\}$ in which each symbol appears exactly once in each row and column.  Each of the squares in \cref{2mols4} is an example of a latin square of order 4. The entry in row~$i$ and column~$j$ of a square $L$ is denoted $L_{ij}$. Two latin squares of the same order, $L$ and~$M$, are said to be \textit{orthogonal} if there is a unique solution $L_{ij} = a$, $M_{ij} = b$ for every pair of $a,b \in \{0,1,\dotsc,n-1\}$.  A set of $k$ latin squares of order~$n$, is called a set of \textit{mutually orthogonal latin squares} (MOLS) if all squares are pairwise orthogonal---in which case we label the system as $\mols{k}{n}$.  \cref{2mols4} shows a set of $\mols{3}{4}$.  The maximum number of mutually orthogonal latin squares of order $n$ is denoted $\N(n)$.   There do not exist $\mols{4}{4}$, so $\N(4) = 3$.  Latin squares, sets of mutually orthogonal latin squares, and a myriad of related objects have numerous applications in statistics, reliability testing, coding theory, and recreational mathematics \cite{colbourn2006handbook,laywine1998discrete}.
\addtolength{\tabcolsep}{2pt}
\renewcommand{\arraystretch}{1}
\begin{figure}[t]
\centering
\begin{tabular}{|cccc|}
\hline
\cellcolor{blue!25}0 & \cellcolor{red!25}1 & \cellcolor{green!25}2 & \cellcolor{yellow!25}3 \\ 
\cellcolor{red!25}1 & \cellcolor{blue!25}0 & \cellcolor{yellow!25}3 & \cellcolor{green!25}2 \\ 
\cellcolor{green!25}2 & \cellcolor{yellow!25}3 & \cellcolor{blue!25}0 & \cellcolor{red!25}1 \\ 
\cellcolor{yellow!25}3 & \cellcolor{green!25}2 & \cellcolor{red!25}1 & \cellcolor{blue!25}0 \\ \hline
\end{tabular}
\quad
\begin{tabular}{|cccc|}
\hline
\cellcolor{blue!25}0 & \cellcolor{red!25}1 & \cellcolor{green!25}2 & \cellcolor{yellow!25}3 \\ 
\cellcolor{green!25}2 & \cellcolor{yellow!25}3 & \cellcolor{blue!25}0 & \cellcolor{red!25}1 \\ 
\cellcolor{yellow!25}3 & \cellcolor{green!25}2 & \cellcolor{red!25}1 & \cellcolor{blue!25}0 \\ 
\cellcolor{red!25}1 & \cellcolor{blue!25}0 & \cellcolor{yellow!25}3 & \cellcolor{green!25}2 \\ \hline
\end{tabular}
\quad
\begin{tabular}{|cccc|}
\hline
\cellcolor{blue!25}0 & \cellcolor{red!25}1 & \cellcolor{green!25}2 & \cellcolor{yellow!25}3 \\ 
\cellcolor{yellow!25}3 & \cellcolor{green!25}2 & \cellcolor{red!25}1 & \cellcolor{blue!25}0 \\ 
\cellcolor{red!25}1 & \cellcolor{blue!25}0 & \cellcolor{yellow!25}3 & \cellcolor{green!25}2 \\ 
\cellcolor{green!25}2 & \cellcolor{yellow!25}3 & \cellcolor{blue!25}0 & \cellcolor{red!25}1 \\ \hline
\end{tabular} \\[0.5\baselineskip]
\caption{A set of three mutually orthogonal latin squares of order four.\label{2mols4}}
\end{figure}
\renewcommand{\arraystretch}{1}
\addtolength{\tabcolsep}{-2pt}

A pair of orthogonal latin squares of order 4 were known to medieval Muslim mathematicians.  In 1700, the Korean mathematician Choi Seok-jeong presented $\mols{2}{9}$ and reported being unable to find a pair of order 10 \cite{colbourn2006handbook}.  Euler could construct a pair of orthogonal latin squares for all odd $n$ and for $n$ divisible by 4 \cite{euler_recherches_1782}. He knew that $\mols{2}{2}$ was impossible and failed to construct a pair of order 6 using any of the methods for which he had success for other $n$.  He verified that none of his construction methods could succeed when $n \equiv 2 \pmod 4$  and conjectured that $\mols{2}{n}$ exist if and only if $n \not\equiv 2 \pmod 4$.  In 1896, E. H. Moore proved that $\mols{(q-1)}{q}$ exist for every prime power $q$ and that if $\mols{k}{n}$ and $\mols{k}{m}$ exist then $\mols{k}{mn}$ exist \cite{MR1505714}. This proves that $\N(n) \geq q-1$ where $q$ is the smallest prime power in the prime power decomposition of $n$.  Tarry confirmed that $\mols{2}{6}$ did not exist in 1900 \cite{tarry_problem_1900}.  In 1949, Bruck and Ryser showed that $\N(n) < n-1$ if $n \equiv 1,2 \pmod 4$ and~$n$ cannot be written as a sum of two squares \cite{MR27520}. Bose and Shrikhande refuted Euler's Conjecture by constructing $\mols{2}{22}$ in 1958 \cite{MR104590}.  Two years later, Bose, Shrikhande, and Parker proved that $\mols{2}{n}$ exist for all $n \equiv 2 \pmod 4$ with $n=2,6$ being the only exceptions \cite{MR122729}.  The only values of $n$ for which it is known that $\N(n) = n-1$ are prime powers \cite{colbourn2006handbook}.  In one of the most famous computational results in combinatorics,  Lam, Thiel, and Swiercz showed that $\N(10) < 9$ \cite{MR1018454}.  This is the only known case where $\N(n) < n-1$ which is not ruled out by the Bruck--Ryser--Chowla Theorem \cite{colbourn2006handbook}.  It is known that $\N(n) \geq n^{1/14.8}$ for $n$ sufficiently large \cite{MR732823}.  The Handbook of Combinatorial Designs contains tables of the best known bounds on $\N(n)$ in 2007 \cite{colbourn2006handbook}.  It is an open problem if $\mols{3}{10}$ exist although McKay, Meynert, and Myrvold have eliminated the possibility that such a triple could have a non-trivial automorphism group \cite{MR2291523}.

Many mathematical methods for constructing sets of mutually orthogonal latin squares exist.  These range from algebraic to the recursive \cite{colbourn2006handbook}.  On the computational side, there have been exhaustive searches and non-exhaustive search using various metaheuristics.  Lam, Thiel, and Swiercz's proof showing the non-existence of $\mols{9}{10}$ was an exhaustive backtracking search aided by powerful theorems from coding theory and permutation groups to eliminate finding structures isomorphic to those already found \cite{MR1018454}. McKay, Meynert, and Myrvold use the orderly generation method to enumerate all non-isomorphic latin squares of orders up to 9 and all latin squares of order 10 with non-trivial automorphism groups \cite{MR2291523}.  They use the library {\tt nauty} \cite{nauty} to compute automorphisms and canonical representatives in each class. They also enumerate the different numbers of equivalence classes of these squares.  McKay and Wanless are able to enumerate all latin squares of order 11 \cite{MR2176596} at the cost of not computing the additional data reported by McKay, Meynert, and Myrvold.  They proceed by generating the squares row by row using an algorithms of Sade \cite{MR2176596,MR0027248}. Niskanen and \"{O}sterg\aa rd's clique finding software, {\tt cliquer} has been successfully used to find mutually orthogonal latin cubes \cite{MR3424339,cliquer}.
Kidd~\cite{kidd2012existence} and Benad\'e, Burger, and van Vuuren~\cite{benade2013enumeration} use custom-written backtracking searches to enumerate all triples of MOLS up to order 8.
Using satisfiability (SAT) solvers with an interface to {\tt nauty}, the nonexistence of $\mols{9}{10}$ was confirmed by Bright \etal by producing nonexistence proof certificates \cite{bright2021sat}.

Colbourn and Dinitz~\cite{MR1398191} implemented many of the combinatorial constructions to aid generating the tables given in the Handbook of Combinatorial Designs~\cite{colbourn2006handbook}; strictly speaking, these computational methods are not metaheuristics but neither are they exhaustive searches.  Much of Colbourn and Dinitz's work has recently been implemented in the SageMath open source mathematical software system \cite{sagemath}.  Elliott and Gibbons~\cite{MR1165804} use the {\em simulated annealing} metaheuristic to construct latin squares of orders up to 18 which contain no subsquares.  Magos~\cite{MR1373306} has used the Tabu metaheuristic to construct latin squares. Mariot \etal~\cite{MR3756068} have used evolutionary algorithms and cellular automata to generate orthogonal latin squares.

An alternative approach to constructing latin squares is to
formulate the problem in a declarative way and use an automated reasoning solver
to search for feasible solutions.  For example, Moura~\cite{Moura1996,Moura1999}
uses integer programming (IP) to search for certain combinatorial designs that are related to latin squares.
Huang \etal~\cite{huang2019investigating} use constraint programming (CP) to search for orthogonal golf designs.
FeiFei and Jian~\cite{ma2013finding} use a first-order logic model generator to search for orthogonal latin squares while
Zaikin and Kochemazov~\cite{zaikin2015search} use a propositional logic solver (i.e., SAT solver)
to search for orthogonal diagonal latin squares.
Gomes, Regis, and Shmoys~\cite{gomes2004improved} employ a hybrid CP/IP approach to the problem of completing
partial latin squares and Appa, Mourtos, and Magos~\cite{appa2006searching,CP2002} investigate using IP
and CP algorithms for generating pairs and triples of mutually orthogonal latin squares.
They report encouraging results and propose that a hybrid IP/CP strategy of integrating the two
techniques might have some success at finding $\mols{3}{10}$ if such a set exists.
General background for IP and CP including Appa \etal's work is reviewed in \cref{background}
and the IP/CP models are outlined in \cref{model_imp}.

In order to explore the feasibility of the IP/CP approach on modern software and hardware
we implement IP and CP models for the mutually orthogonal latin squares problem
and examine the performance of solvers run on those models.
Using out-of-the-box solvers on a single desktop computer we
find mutually orthogonal latin squares (or disprove their existence)
in all orders up to ten.
We improve the performance of the solver by formulating
an improved constraint programming formulation (described in \cref{subsec:cp_model}) and
an extended form of symmetry breaking (described in \cref{subsec:symm_break}).
We give detailed timings for each method (see \cref{sec:results}) and
also apply the method to search for triples of orthogonal latin squares
(see \cref{3mols_timings}).
Finally, we compare our timings with previously published timings (see \cref{sec:comparison})
and comment on discrepancies that we observed between our results and prior results.

\section{Background}\label{background}

Integer linear programming (IP) and constraint programming (CP) are
two paradigms for solving combinatorial optimization and search problems.
IP solvers depend heavily on numerical solutions to linear
programming problems whereas CP solvers are based on search techniques
such as domain reduction. 
The two approaches offer complementary strengths
and there has been a significant amount of research on integrating the
two to solve difficult combinatorial optimization search problems~\cite{hooker}.

Both approaches use the ``divide and conquer'' paradigm of splitting a problem into subproblems by
choosing a variable and then branching on each possible value for that variable.  Each subproblem
formed in this way is called a \emph{node} of the search.  A \emph{branch} of the search is formed
by assigning a fixed value to the branching variable selected in each node.
IP and CP solvers typically report the number of nodes and/or the number of branches that they
explore during the search.

IP solvers also solve the linear programming ``relaxation'' of the problem by
dropping the constraint that the variables must be integers.  If the relaxation has no solution over
the rationals then the original problem must also be infeasible.  If the relaxation has a solution
then the solver examines if it violates any inequalities that can be derived from the original
set of linear constraints assuming that the variables must take on integer values.  Such inequalities
are known as \emph{cutting planes} and they can be effective at pruning the search space and forcing the
solution of the linear program to integer values.  If no cutting planes can be derived and the
relaxation solution has non-integral values then a variable is chosen to branch on and new nodes are
created.

Intuitively, a CP solver is effective at enumerating solutions
quickly as it is not required to solve linear programs during its search.  However,
the relaxation solutions provided by the IP solver---despite being relatively expensive to compute---%
help provide a ``global'' perspective of the search space.
In particular, Appa \etal present IP and CP models to search for MOLS and
develop sophisticated techniques for combining solvers in ways
that exploit each solver's strengths~\cite{appa2006searching,CP2002}.

In addition to pure IP/CP approaches they present ways of embedding CP into IP and vice versa.
They are able to improve on an ``out-of-the-box'' solver using these hybrid approaches in each
case when searching for $\mols{2}{n}$ with $n\leq12$.  They also extend their method to searching
for triples of MOLS---using a CP solver to complete a single latin square and an IP solver to
search for a pair of MOLS orthogonal to the first square.  If the IP solver determines such a
pair does not exist then the CP solver finds a new latin square and the process repeats.
It is also possible for the CP solver to pass only a partially-completed first square to the IP
solver.  This makes the CP solver's search easier at the cost of making the IP solver's search
more difficult and Appa \etal provide timings demonstrating that the runtime is typically minimized
by passing a half-completed first square to the IP solver.

Moreover, Appa \etal analyze the symmetries present in the $\mols{k}{n}$ problem and provide effective
methods of removing much of the symmetry.  This is important
in order to prevent solvers from exploring
parts of the search space which are actually the same under some symmetry.
We discuss their symmetry breaking method in detail in \cref{subsec:symm_break} and derive
an extension that we use in our work.

\section{Models}\label{model_imp}

In this section we describe the models that we use when searching for MOLS
using integer programming (see \cref{subsec:ip_model}) and constraint programming (see \cref{subsec:cp_model})
as well as the symmetry breaking methods that we use (see \cref{subsec:symm_break}).

\subsection{Integer Programming Model}\label{subsec:ip_model}

One straightforward method of approaching the $\mols{2}{n}$ problem is to express it as a pure binary linear program and use integer programming (IP) solvers to generate solutions~\cite[\S26.3.IV]{dantzig1963linear}.
Let $X$ and $Y$ be two mutually orthogonal latin squares of order $n$.
Our IP model for a set of two mutually orthogonal latin squares contains the $n^4$ binary variables
\[ x_{ijkl} \coloneqq
\begin{cases} 
      1 & \text{if $k$ and $l$ appear in cells $(i,j)$ of squares $X$ and $Y$ respectively} \\
      0 & \text{otherwise} 
   \end{cases}
\]
for all $ i,j,k,l \in \{0,1,\dotsc,n-1\}$.
 
We encode the latin and orthogonality constraints as six sets of $n^2$ equalities grouped by which subscripts of $x_{ijkl}$ are fixed: 
\begin{align*}
\sum_{0\leq k,l < n} x_{ijkl} &= 1 \: \forall i,j  && \text{each cell contains only 1 value} \\
\sum_{0\leq j,l < n} x_{ijkl} &= 1 \: \forall i,k  && \text{latin property in rows of $X$} \\
\sum_{0\leq j,k < n} x_{ijkl} &= 1 \: \forall i,l  && \text{latin property in rows of $Y$} \\
\sum_{0\leq i,l < n} x_{ijkl} &= 1 \: \forall j,k  && \text{latin property in columns of $X$} \\
\sum_{0\leq i,k < n} x_{ijkl} &= 1 \: \forall j,l  && \text{latin property in columns of $Y$} \\
\sum_{0\leq i,j < n} x_{ijkl} &= 1 \: \forall k,l  && \text{orthogonality of $X$ and $Y$} 
\end{align*}
This model can be extended to search for $\mols{k}{n}$ using $n^{k+2}$ variables and $\binom{k+2}{2}n^2$ constraints, but this easily exceeds our ability to solve for all but the smallest orders of~$k$.
Instead, in order to seach for triples of MOLS we employ an alternative method described
by Appa \etal~\cite{appa2006searching} that relies on using a separate solver to fill in
the entries of a third square~$Z$.
Assuming the entries of $Z$ have been fixed in advance,
the $\mols{3}{n}$ problem may then be represented
as an extension of the above $\mols{2}{n}$ model
with the following additional constraints:
\begin{align*}
\sum_{\substack{Z_{ij} = z \\ 0\leq l < n}} x_{ijkl} &=1 \: \forall k,z  && \text{orthogonality of $X$ and $Z$} \\
\sum_{\substack{Z_{ij} = z \\ 0\leq k < n}} x_{ijkl} &=1 \: \forall l,z  && \text{orthogonality of $Y$ and $Z$}
\end{align*}
 
\subsection{Constraint Programming Model}\label{subsec:cp_model}

A CP solver allows a more natural formulation of the $\mols{2}{n}$ problem using the $2n^2$ integer-valued variables
\begin{align*}
X_{ij} &\coloneqq \text{value of cell $(i,j)$ in square $X$} , \\
Y_{ij} &\coloneqq \text{value of cell $(i,j)$ in square $Y$} ,
\end{align*}
where $i,j,X_{ij},Y_{ij} \in \{ 0, 1, \dotsc , n-1 \}$.
The squares $X$ and $Y$ can be forced to be latin squares via ``AllDifferent'' constraints
which specify that the rows and columns of the squares each contain different values.
One natural way of encoding orthogonality (c.f.~\cite{appa2006searching}) is to
use a set of $n^2$ variables $Z_{ij}$ defined by the linear constraints
\[ Z_{ij} = X_{ij} + n Y_{ij} \in \{0,1,\dotsc,n^2-1\} . \]
The squares $X$ and $Y$ are orthogonal exactly when the variables $Z_{ij}$ contain no duplicate values and
therefore the single constraint
$\text{AllDifferent}(Z_{ij} \: \forall i,j)$ ensures that $X$ and $Y$ are orthogonal.

We also used an alternative formulation of ensuring orthogonality that uses
no arithmetic and shorter AllDifferent constraints.
In order to describe it, note that each row of a latin square may be considered as a permutation of $\{0,\dotsc,n-1\}$.
Denote by $AB$ the square whose $i$th row consists of the composition of the permutations $A_i$ and $B_i$ (the $i$th rows of $A$ and $B$).
Composition is performed left-to-right following the standard convention in the field,\footnote{The standard reference~\cite{laywine1998discrete}
defines function composition as performed right-to-left but in fact follows the left-to-right convention.} i.e., $fg$ denotes
the function $x\mapsto g(f(x))$ and so the $(i,j)$th entry of $AB$ is $B_i(A_i(j))$.
Additionally, $A^{-1}$ denotes the square where each row consists
of the inverse of the permutation formed by the corresponding row of $A$.

Our alternative orthogonality encoding is based on~\cite[Theorem 6.6]{laywine1998discrete}
which implies that two latin squares $X$ and $Y$ are orthogonal if and only if there is a latin square $Z$ such that $XZ=Y$.
The additional variables in our alternative orthogonality encoding are
\begin{align*}
Z_{ij} &\coloneqq \text{value of cell } (i,j) \text{ in square } Z = X^{-1}Y ,
\end{align*}
where $Z_{ij} \in \{ 0,1, \dotsc ,n-1 \} $.  In order to ensure $Y=XZ$
the $(i,j)$th entry of $Y$ is set equal to the $(i,X_{ij})$th entry of $Z$.  This
is accomplished with the ``element indexing'' constraint $Z_i[X_{ij}] = Y_{ij}$ where $Z_i$ is the vector of
variables corresponding to the $i$th row of $Z$.
Some CP solvers such as OR-Tools~\cite{ortools}
have native support for such constraints via what are called ``VariableElement'' constraints.
The constraints encoding that the squares $X$ and $Y$ are orthogonal are then
$\text{AllDifferent}(Z_{ij} \: \forall j)$ and $\text{AllDifferent}(Z_{ij} \: \forall i)$.
Altogether this encoding uses $6n$ AllDifferent constraints together
with $n^2$ element indexing constraints.
This encoding also may be straightforwardly extended to encode the $\mols{3}{n}$ problem
by adding another $n^2$ variables encoding the entries of a third latin square $U$
and using the same orthogonality constraints as above between $X$ and $U$ and $Y$ and $U$
(requiring another $2n^2$ variables $Z_{ij}^{XU}$ and $Z_{ij}^{YU}$).

\subsection{Symmetry Breaking}\label{subsec:symm_break}

There are a large number of possible symmetries in the $\mols{k}{n}$ problems.  In any set of $\mols{k}{n}$, the squares themselves can be permuted, the rows of all squares can be permuted simultaneously, the columns of all the squares can be permuted simultaneously, the symbol sets within each square can be permuted independently and the squares may all be replaced with their transposes. All of these operations preserve the latin properties and orthogonality of the set of squares \cite{colbourn2006handbook}.  This group of possible symmetries has order $2k!(n!)^{k+2}$ and the search space can be reduced significantly for any elimination of possible symmetries that still permits finding an isomorphic representative of any solution.  Appa \etal fix the first row of every square to be in lexicographic order which eliminates the permutations of the rows and fixes the symbol permutations to be the same in each square.  They fix the column of the first square to be in lexicographic order which eliminates permutations of the rows.  This reduces the possible symmetries to $2k!n!$. With these cells fixed, the first column of $Y$ must be a permutation where $Y_{i0} \neq i$ for $1 \leq i < n$.  The number of such permutations is approximately $(n-1)!/e$ \cite{MR1971432}. Appa \etal~\cite{appa2006searching,CP2002} show every possible solution is isomorphic to one satisfying $Y_{10} = 2$, $Y_{i0} \neq i$ and $Y_{i0} \leq i+1$ for $1 \leq i < n$. To understand the magnitude of this reduction of the search space we must count the number of these solutions. Let $P(n)$  be the set of all permutations of $\{0,1,\ldots,n-1\}$ and 
\[
  A(n) = \{\rho \in P(n): \rho(0)= 0,\; \rho(1) =2,\; \rho(i) \neq i,\; \rho(i) \leq i+1\},
\]
and $a(n) = |A(n)|$.
\begin{proposition} Let $F(n)$ be the $n$th Fibonacci number.  Then
  $
    a(n) = F(n-2).
  $
  \end{proposition}
  See the appendix for a proof.
    With this reduction in cases, Appa \etal have reduced the number of first columns for $Y$ from $(n-1)!/e$ to
    \[
  a(n) = F(n-2) \approx \frac{\sqrt{5}}{5}\biggl ( \frac{1+\sqrt{5}}{2} \biggr )^{n-2}
\]
choices.  This is a very significant reduction.

By exploiting the disjoint cycle structure of these permutations, we have been able to reduce this number to $b(n)$, the number of partitions of $n-1$ into parts of size greater than 1.  The values of $b(n)$ are sequence A002865 at the Online Encyclopedia of Integer Sequences and the value of $b(n)$ is approximately \cite{A002865}
\[
b(n) \approx \frac{\pi  e^{\pi\sqrt{2(n-1)/3}}}{12\sqrt{2}(n-1)^{3/2}}.
\]
This growth rate is substantially slower than the Fibonacci numbers. The values of $b(n)$ for $n\leq12$ are:
  \[
  \begin{array}{c@{\quad}c@{\quad}c@{\quad}c@{\quad}c@{\quad}c@{\quad}c@{\quad}c@{\quad}c@{\quad}c@{\quad}c}
    n & 3 & 4 & 5 & 6 & 7 & 8 & 9 & 10 & 11 & 12 \\ \hline
    b(n) & 1 & 1 & 2 & 2 & 4 & 4 & 7 & 8 & 12 & 14
  \end{array}
  \]
  Let $X$ and $Y$ be a pair of orthogonal latin squares of order $n$ such that the first rows of both are in lexicographic order and the first column of $X$ is in lexicographic order.  The {\em first column permutation } of $X$ and $Y$ is the permutation of the symbols $\{1, 2, \ldots n-1\}$ defined by $\rho(i) = Y_{i0}$. The {\em ordered cycle} of $X$ and $Y$ is  the representation of $\rho$ as a product of disjoint cycles
  \[
    (a_{0,0},a_{0,1},\ldots,a_{0,l_0-1})(a_{1,0},\ldots,a_{1,l_1-1})\cdots(a_{c,0},\ldots,a_{c,l_c-1})
    \]
    where each cycle is the lexicographic least cyclic shift and the cycles are in lexicographic order \cite{MR1409812}.
    The {\em ordered cycle type} of $X$ and $Y$ is $(l_0,l_1,\ldots,l_c)$.
    \begin{theorem} \label{conj}
      Let $X$ and $Y$ be a pair of orthogonal latin squares.  They are isomorphic to a pair $X'$, $Y'$ where the first row of both is in lexicographic order, the first column of $X'$ is in lexicographic order and the ordered cycle type of $X'$ and $Y'$ is non-decreasing.
\end{theorem}
See the appendix for a proof.
This theorem applies equally well to a set of $\mols{k}{n}$.

\section{Results}\label{sec:results}

In this section we provide running times for our implementations.
All times were recorded using a single thread on an
Intel Core i9-9900K processor running at 3.6~GHz and with 32~Gb of memory.
The models were compiled into executable code using Microsoft Visual C++ 2019.
The IP models were solved using Gurobi version 9~\cite{gurobi_sym},
the CP models were solved using OR-Tools version 8~\cite{ortools},
and all default solver parameters were used with the exception of
disabling multi-threading.

By default OR-Tools uses a branching strategy that chooses the first
unassigned variable.  We declared the variables so the values of the
entries of $X$ and $Y$ are ordered by row and column
(i.e., in the order $X_{00}$, $Y_{00}$, $X_{01}$, $Y_{01}$, $\dotsc$)
and the variables ensuring orthgonality are ordered last.
Our code is publicly available at \href{https://github.com/noahrubin333/CP-IP}{github.com/noahrubin333/CP-IP}.
We first describe the times for searching for a pair of MOLS and then
describe our experiments searching for a triple of MOLS.

\subsection{$\mols{2}{n}$ Timings}\label{2mols_timings}

As a ``baseline'' we first present timings for the models without
any symmetry breaking included.
In these instances the solver either reported infeasible (in order six), explicitly found a pair of MOLS,
or timed out after 60,000 seconds.
In \cref{no_sym_timings} we compare the IP model along with
the two CP models---one with linear constraints (CP-linear)
and one with indexing constraints (CP-index).
The timings in the table show that the CP-index model generally performs
the best with the exception of $n=6$ where Gurobi performs much better.

The order six case is exceptional since it is the only infeasible case
among the cases described in \cref{no_sym_timings}.  In this case the solver
has to examine the entire search space in order to rule out the existence
of $\mols{2}{6}$ and in such a case symmetry breaking is particularly
important.  We found that Gurobi performed exceptionally well in this
case as it has some symmetry breaking routines that automatically detected
and removed symmetry present in the search space~\cite{gurobi_sym}.

Due to the superiority of the CP-index model over the CP-linear model in
the remaining timings we only present timings for the CP-index model.
In \cref{ipcp_timings} we present detailed timings for the IP and CP-index
models with the ``dominance detection'' symmetry breaking of Appa \etal
and the ``cycle type'' symmetry breaking method described in
\cref{subsec:symm_break}.

\begin{table}[t]
\centering
\begin{tabular}{c@{\;\quad}r@{\quad}r@{\quad}r@{\quad}r@{\quad}r@{\quad}r}
Model & \df{5} & \df{6} & \df{7} & \df{8} & \df{9} & \df{10} \\ \hline
IP & 0.09  & 2.4  & 1.3  & 6.0  & 284.5  & 9,502.9  \\
CP-linear & 0.03  & Timeout & 3.4  & 81.9  & 1,517.0  & 12,515.4  \\
CP-index & 0.03  & 49,998.4  & 3.6  & 28.2  & 328.4  & 189.6 
\end{tabular}  \\[0.5\baselineskip]
\caption{Timings (in seconds) for orders $5\leq n\leq10$ without symmetry breaking;
all models timed out at 60,000 seconds for $n=11$.
\label{no_sym_timings}}
\end{table}

\begin{table}[t]
\centering
\begin{tabular}{c@{\;\,}l@{\quad}r@{\quad}r@{\quad}r@{\quad}r}
$n$ & \df{Sym.~Breaking} & \df{IP Time} & \df{CP Time} & \df{IP Nodes} & \df{CP Branches} \\ \hline
6 & Dominance & 2 & 0 & 1,551 & 5,048 \\ 
6 & (1, 5) & 0 & 0 & 0 & 2,286 \\
6 & (1, 2, 3) & 0 & 0 & 0 & 2,312 \\
7 & Dominance & 6 & 0 & 2,557 & 4,981 \\ 
7 & (1, 2, 2, 2) & 0 & 0 & 102 & 4,255 \\
7 & (1, 6) & 0 & 0 & 102 & 3,686 \\
7 & (1, 3, 3) & 0 & 0 & 102 & 7,508 \\
7 & (1, 2, 4) & 0 & 0 & 102 & 9,178 \\
8 & Dominance & 23,729 & 12 & 3,159,370 & 291,275 \\ 
8 & (1, 3, 4) & 4 & 5 & 107 & 140,352 \\
8 & (1, 2, 2, 3) & 5 & 12 & 271 & 245,779 \\
8 & (1, 2, 5) & 4 & 8 & 106 & 177,662 \\
8 & (1, 7) & 5 & 8 & 121 & 170,305 \\
9 & Dominance & Timeout & 110 & 4,932,210 & 2,336,846 \\ 
9 & (1, 3, 5) & 299 & 509 & 7,032 & 7,823,017 \\
9 & (1, 2, 2, 4) & 1,133 & 546 & 14,917 & 5,724,059 \\
9 & (1, 8) & 504 & 292 & 9,106 & 4,554,707 \\
9 & (1, 2, 6) & 460 & 366 & 8,419 & 6,000,262 \\
9 & (1, 2, 2, 2, 2) & 714 & 39 & 11,656 & 912,110 \\
9 & (1, 4, 4) & 321 & 632 & 7,021 & 9,197,213 \\
9 & (1, 2, 3, 3) & 753 & 134 & 11,827 & 2,570,887 \\
10 & Dominance & Timeout & 1,486 & 3,302,400 & 24,843,924 \\ 
10 & (1, 2, 2, 5) & 14,256 & 2,312 & 96,111 & 30,658,692 \\
10 & (1, 2, 2, 2, 3) & 67,767 & 2,869 & 442,145 & 39,949,251 \\
10 & (1, 9) & 16,339 & 2,375 & 108,529 & 31,813,531 \\
10 & (1, 2, 3, 4) & 88,183 & 21,194 & 551,128 & 191,728,592 \\
10 & (1, 3, 6) & 57,652 & 5,169 & 361,129 & 72,848,024 \\
10 & (1, 2, 7) & 19,330 & 423 & 108,138 & 5,956,008 \\
10 & (1, 3, 3, 3) & 53,834 & 1,297 & 368,203 & 17,850,420 \\
10 & (1, 4, 5) & 22,470 & 5,205 & 160,559 & 74,229,870
\end{tabular}  \\[0.5\baselineskip]
\caption{Detailed IP and CP timings (in seconds) by symmetry breaking type.\label{ipcp_timings}}
\end{table}

For the ``cycle type'' symmetry breaking method the first column of $X$ was fixed in sorted order
and a random assignment of the first column of $Y$
was assigned---only subject to the constraint that $Y_{i0}\neq i$ for $i>0$.
In the IP model the appropriate variables were assigned to $1$ to encode
the values in the first columns of $X$ and~$Y$.
In the CP model the variables in the first column
were set by restricting their domains to a single element.

A Python script was used to determine the cycle type of the first column
permutation of $Y$ and then categorize that instance into one of the
$b(n)$ cycle type equivalence classes described in \cref{subsec:symm_break}.
In each order, 100 independent runs (with the exception of $n=10$ which used 10 runs)
with no timeout were made in order to determine a ``typical'' running time for each cycle type.
\cref{ipcp_timings} shows the result of the median run for each cycle type when the runs were sorted
in terms of running time.  In the IP model the cycle type
symmetry breaking method performed much better
than the dominance detection method.  In the CP model
the results were highly variable: some runs
completed much faster using dominance detection
and other runs completed much slower.

\subsection{$\mols{3}{n}$ Timings}\label{3mols_timings}
Lastly, we give timings for the $\mols{3}{n}$ models presented in \cref{model_imp}.
To implement these models we used the CP solver to enumerate latin squares of order $n$ and these are taken as our third square.
For each square we generate the constraint sets to encode that the third square is orthogonal to the first two squares and then solve the
model using an IP or CP solver.
If the model returns a feasible solution then we have found a solution to $\mols{3}{n}$,
otherwise we continue by replacing the aforementioned constraints with ones generated from the next latin square produced by the CP solver.
Symmetry breaking was disabled except for specifying that the first row of each square occurs in sorted order.
The average time it took the solver to solve a single instance with the third square fixed is given in \cref{times_3MOLS}.

A $\mols{3}{7}$ was found after 11,785 attempts and a $\mols{3}{8}$ was found after a single attempt.
We observed that the very first latin square of order 8 generated by the CP solver
(using the variable ordering described in \cref{sec:results}) is highly structured and
part of a mutually orthogonal triple---this may explain why the timings of Appa \etal~\cite{appa2006searching,CP2002}
were particularly fast in order 8.  See the appendix for the $\mols{3}{8}$ that we found with the
highly structured third square.

\begin{table}[t]
\centering
\begin{tabular}{c@{\quad}r@{\quad}r@{\quad}r@{\quad}r@{\quad}r@{\quad}r}
Model & \df{5} & \df{6} & \df{7} & \df{8} & \df{9} & \df{10} \\ \hline
IP       & 0.0 & 0.5 & 4.2 & 82.0 & Timeout & Timeout \\
CP-index & 0.0 & 0.1 & 0.2 & 0.4 & 6.3 & 1470.3
\end{tabular}  \\[0.5\baselineskip]
\caption{Average time (in seconds) for solving a $\mols{3}{n}$ instance with given order~$n$ and the third square fixed;
the timeout was set to 50,000 seconds.\label{times_3MOLS}}
\end{table}

\section{Comparison with Prior Work}\label{sec:comparison}

Table~\ref{tbl:compare} compares our times with those of Appa, Magos, and Mourtos~\cite{appa2006searching,CP2002}.
The times are not directly comparable since Appa \etal
perform their timings on a Pentium-III processor at 866~MHz with 256~Mb of RAM
and make a number of custom modifications to the behaviour of their solvers.
However, they also provide the running
times using a ``baseline'' version of the IP solver
XPRESS-MP (i.e., with no modifications or customizations).
\begin{table}[t]
\centering
\begin{tabular}{c@{\quad}r@{\quad}r@{\quad}r@{\quad}r}
& \multicolumn{2}{c}{Appa \etal~\cite{CP2002}} & \multicolumn{2}{c}{Our work} \\
$n$ & IP Time & Nodes/sec & IP Time & Nodes/sec \\ \hline
5 & 12 & 0 & 0 & 11 \\
6 & 1,843 & 16 & 2 & 500 \\
7 & 7,689 & 3 & 2 & 76 \\
8 & 294 & 4 & 6 & 18 \\
9 & 32,341 & 10 & 284 & 22 \\
10 & 37,065 & 154 & 9,503 & 6 \\
11 & 46,254 & 3,333 & Timeout & 4 \\
12 & 59,348 & 11,293 & Timeout\rlap{$^*$} & 20
\end{tabular}  \\[0.5\baselineskip]
\caption{Comparison of the ``baseline'' times~\cite{CP2002}
and our IP times with no symmetry breaking enabled.
For $n=12$, Gurobi's heuristic search instantly found a pair of MOLS.
This was disabled to determine the node processing time;
the timeout was set to 60,000 seconds.
\label{tbl:compare}
}
\end{table}
These baseline timings should make it possible to compare how
far software and hardware has come since 2002.
Surprisingly, we find that the running
times reported by Appa \etal for $n>10$ are better than our timings (see \cref{tbl:compare}).
In particular, our processing time per node tends to increase in the order while their
processing time per node seems to decrease exponentially in the order.%
\footnote{The number of nodes processed as reported in~\cite{CP2002}
is not consistent with the data provided elsewhere~\cite{appa2006searching,mourtos2003integer}.
Regardless, the node processing time in every case is reported to increase exponentially.}
By personal communication~\cite{mourtos_pc} we have been informed
the observed times may be the result of a
restriction of the search space that was undocumented and no longer known to the authors.

The code of Appa \etal is not available for review, but the analysis
reported by Mourtos~\cite[\S6.5.3]{mourtos2003integer} does suggest the search space
was restricted in some way.  In order to estimate the running time to perform an exhaustive
search for $\mols{3}{10}$, Mourtos reports the running time of a hybrid CP/IP model where the
CP solver enumerates solutions of the third square and the IP solver attempts to find another
$\mols{2}{10}$ orthogonal to the third square.  After randomly fixing zeros in the third square
the combined CP/IP model reported infeasibility for the specific fixing of zeros in an average of 45,134 seconds.
However, there are over $7.5\cdot10^{24}$ possibilities~\cite{A000315}
for the final square (accounting for the fixed first row and fixed zeros)
that the CP solver would have had to enumerate---indicating that the search space must have
been further restricted somehow for the CP solver to finish its enumeration in a reasonable time.

\section{Conclusion}\label{sec:conclusion}

In this paper we use integer and constraint programming solvers in the search for mutually orthogonal
latin squares.  We demonstrate that modern state-of-the-art solvers are able to 
find pairs of MOLS (or demonstrate their nonexistence) in a reasonable amount of time for all orders $n\leq10$
but struggle with higher orders.
We continue the work of Appa \etal by extending their symmetry breaking method and providing an alternative
CP formulation that performs better in our implementation.

We also examine the feasibility of using CP/IP solvers in the search for triples of MOLS.
One of the motivations of this work was to examine the feasibility of using CP and IP solvers
in the search for a triple of MOLS of order ten.  We were encouraged by the estimations presented
by Mourtos~\cite{mourtos2003integer}
who fixed a pattern of zeros in one of the squares and reported that
it is possible to determine the nonexistence of a triple of MOLS containing that specific
pattern of zeros in about 12.5 hours.
After applying symmetry breaking there are $9!$ ways of fixing the zeros in one of the squares
and this results in an estimate of about 520 years of computing time in order to resolve the question of the
existence of $\mols{3}{10}$.  Unfortunately, our results find this estimate to be much too
optimistic.  Based on personal communication~\cite{mourtos_pc}
and a careful analysis of previous work we conclude that
previous results included additional restrictions of the search space.
We were unable to rule out even a single fixing of zeros from
appearing in a $\mols{3}{n}$.  Our CP implementation required an average of about 25 minutes to
show that a fixed latin square could not be part of a $\mols{3}{n}$.  Given that
there are about $7.5\cdot10^{24}$ possibilities for a fixed latin
square (after symmetry breaking) this method would require an estimated $10^{20}$ years
to resolve the $\mols{3}{n}$ existence question---%
clearly impractical, but useful as a
challenging benchmark for IP and CP solvers.

\paragraph{Acknowledgement.}
We thank Yiannis Mourtos for answering many questions concerning his work.

\bibliographystyle{splncs04} 
\bibliography{references}

\newpage

\section*{Appendix}
\setcounter{theorem}{0}
\setcounter{proposition}{0}

\begin{proposition} Let $F(n)$ be the $n$th Fibonacci number.  Then
  $
    a(n) = F(n-2).
  $
  \end{proposition}
  \begin{proof}
    Let $\rho \in A(n)$.  If $\rho(n-1) = n-2$ then $\rho(i) < n-3$ for all $i < n-5$ and $\rho(n-3) \neq n-3$.  Thus $\rho(n-4) = n-3$ and $\rho$ with its last two positions deleted is in $A(n-2)$.  If $\rho(n-1) \neq n-2$, then $\rho(i) < n-2$ for all $i < n-4$ so $\rho(n-3) = n-2$.  In this case $\rho$ with symbol $n-1$ deleted is in $A(n-1)$. From this we can conclude
    \[
      a(n) \leq a(n-1) + a(n-2).
    \]
    Given any $\rho' \in A(n-1)$ we can produce a $\rho \in A(n)$ by inserting symbol $n-1$ in the second last position.  No two $\rho', \rho'' \in A(n-1)$ can produce the same $\rho \in A(n)$ this way and in every case $\rho(n-3) = n-2$. Given any $\rho' \in A(n-2)$ we can produce a $\rho \in A(n)$ by appending $n-1$ followed by $n-2$. No two $\rho', \rho'' \in A(n-2)$ can produce the same $\rho \in A(n)$ this way and in every case $\rho(n-1) = n-2$. Since symbol $n-2$ is in a different position when constructing from $A(n-1)$ or  $A(n-2)$, we can conclude that 
    \[
      a(n) \geq a(n-1) + a(n-2).
    \]
Since $a(3) = a(4) = 1$ we conclude that $a(n) = F(n-2)$. \qed
    \end{proof}

\begin{theorem}
      Let $X$ and $Y$ be a pair of orthogonal latin squares.  They are isomorphic to a pair $X'$, $Y'$ where the first row of both is in lexicographic order, the first column of $X'$ is in lexicographic order and the ordered cycle type of $X'$ and $Y'$ is non-decreasing.
\end{theorem}
\begin{proof}
  We transform $X$ and $Y$ into the desired form by permuting the rows and columns of both squares so that the first row and column of $X$ are in lexicographic order.  Then we permute the symbol set of $Y$ such that the first row of $Y$ is in lexicographic order.  Let $\rho$ be the first column permutation of these transformed $X$ and $Y$ and let $t = (l_0,l_1,\ldots,l_c)$ be the ordered cycle type.  Let $t' = (l_{j_0},l_{j_1},\ldots,l_{j_c})$ be the list $t$ sorted in non-decreasing order. Let
  \[
    r = (a_{j_0,0},a_{j_0,1},\ldots,a_{j_0,l_{j_0}-1},a_{j_1,0},\ldots,a_{j_1,l_{j_1}-1},\ldots,a_{j_c,0},\ldots,a_{j_c,l_{j_c}-1}) 
\]
be a permutation of $\{1, 2, \ldots n-1\}$ in standard list form, not disjoint-cycle form.  Extend $r$ by setting $r(0) = 0$. Apply permutation $r$ to the symbols of both $X$ and $Y$, and call the resulting squares $X'$ and $Y'$ so $X'_{i0} = r(i)$ and $Y'_{i0} = r(Y_{i0})$.  Permute the columns of both $X'$ and $Y'$ so the first rows of each are in lexicographic order.  Now apply $r^{-1}$ to the rows of both $X'$ and $Y'$ so that $X'_{i0} = r(r^{-1}(i)) = i$ and $Y'_{i0} = r(Y_{r^{-1}(i)0}) = r \rho r^{-1}(i)$, the conjugate of $\rho$ by $r$.  Conjugation preserves the cycle structure of a permutation and replaces each element, $x$ of $\rho$ in disjoint cycle notation by $r(x)$ \cite{MR1409812}. Thus
\[
  r \rho r^{-1} = (1,2,\ldots,l_{j_0}-1)(l_{j_0},\ldots,l_{j_1}-1)\cdots(l_{j_{c-1}},\ldots,l_{j_c}-1).
  \]
  And thus the ordered cycle type of $X'$ and $Y'$ is non-decreasing. \qed
  \end{proof}

\addtolength{\tabcolsep}{2pt}
\renewcommand{\arraystretch}{1}
\begin{figure}[t]\centering
\begin{tabular}{|cccccccc|} \hline
\cellcolor{blue!25}0 & \cellcolor{red!25}1 & \cellcolor{green!25}2 & \cellcolor{yellow!25}3 & \cellcolor{cyan!25}4 & \cellcolor{magenta!25}5 & \cellcolor{brown!25}6 & \cellcolor{violet!25}7  \\ 
\cellcolor{cyan!25}4 & \cellcolor{magenta!25}5 & \cellcolor{brown!25}6 & \cellcolor{violet!25}7 & \cellcolor{blue!25}0 & \cellcolor{red!25}1 & \cellcolor{green!25}2 & \cellcolor{yellow!25}3  \\
\cellcolor{brown!25}6 & \cellcolor{violet!25}7 & \cellcolor{cyan!25}4 & \cellcolor{magenta!25}5 & \cellcolor{green!25}2 & \cellcolor{yellow!25}3 & \cellcolor{blue!25}0 & \cellcolor{red!25}1  \\
\cellcolor{green!25}2 & \cellcolor{yellow!25}3 & \cellcolor{blue!25}0 & \cellcolor{red!25}1 & \cellcolor{brown!25}6 & \cellcolor{violet!25}7 & \cellcolor{cyan!25}4 & \cellcolor{magenta!25}5  \\
\cellcolor{yellow!25}3 & \cellcolor{green!25}2 & \cellcolor{red!25}1 & \cellcolor{blue!25}0 & \cellcolor{violet!25}7 & \cellcolor{brown!25}6 & \cellcolor{magenta!25}5 & \cellcolor{cyan!25}4  \\
\cellcolor{violet!25}7 & \cellcolor{brown!25}6 & \cellcolor{magenta!25}5 & \cellcolor{cyan!25}4 & \cellcolor{yellow!25}3 & \cellcolor{green!25}2 & \cellcolor{red!25}1 & \cellcolor{blue!25}0  \\
\cellcolor{magenta!25}5 & \cellcolor{cyan!25}4 & \cellcolor{violet!25}7 & \cellcolor{brown!25}6 & \cellcolor{red!25}1 & \cellcolor{blue!25}0 & \cellcolor{yellow!25}3 & \cellcolor{green!25}2  \\
\cellcolor{red!25}1 & \cellcolor{blue!25}0 & \cellcolor{yellow!25}3 & \cellcolor{green!25}2 & \cellcolor{magenta!25}5 & \cellcolor{cyan!25}4 & \cellcolor{violet!25}7 & \cellcolor{brown!25}6  \\ \hline
\end{tabular}
\quad
\begin{tabular}{|cccccccc|} \hline
\cellcolor{blue!25}0 & \cellcolor{red!25}1 & \cellcolor{green!25}2 & \cellcolor{yellow!25}3 & \cellcolor{cyan!25}4 & \cellcolor{magenta!25}5 & \cellcolor{brown!25}6 & \cellcolor{violet!25}7  \\
\cellcolor{brown!25}6 & \cellcolor{violet!25}7 & \cellcolor{cyan!25}4 & \cellcolor{magenta!25}5 & \cellcolor{green!25}2 & \cellcolor{blue!25}0 & \cellcolor{yellow!25}3 & \cellcolor{red!25}1  \\
\cellcolor{violet!25}7 & \cellcolor{blue!25}0 & \cellcolor{yellow!25}3 & \cellcolor{cyan!25}4 & \cellcolor{magenta!25}5 & \cellcolor{green!25}2 & \cellcolor{red!25}1 & \cellcolor{brown!25}6  \\
\cellcolor{red!25}1 & \cellcolor{brown!25}6 & \cellcolor{magenta!25}5 & \cellcolor{green!25}2 & \cellcolor{blue!25}0 & \cellcolor{cyan!25}4 & \cellcolor{violet!25}7 & \cellcolor{yellow!25}3  \\
\cellcolor{magenta!25}5 & \cellcolor{cyan!25}4 & \cellcolor{violet!25}7 & \cellcolor{brown!25}6 & \cellcolor{red!25}1 & \cellcolor{yellow!25}3 & \cellcolor{blue!25}0 & \cellcolor{green!25}2  \\
\cellcolor{yellow!25}3 & \cellcolor{green!25}2 & \cellcolor{red!25}1 & \cellcolor{blue!25}0 & \cellcolor{violet!25}7 & \cellcolor{brown!25}6 & \cellcolor{magenta!25}5 & \cellcolor{cyan!25}4  \\
\cellcolor{green!25}2 & \cellcolor{magenta!25}5 & \cellcolor{brown!25}6 & \cellcolor{red!25}1 & \cellcolor{yellow!25}3 & \cellcolor{violet!25}7 & \cellcolor{cyan!25}4 & \cellcolor{blue!25}0  \\
\cellcolor{cyan!25}4 & \cellcolor{yellow!25}3 & \cellcolor{blue!25}0 & \cellcolor{violet!25}7 & \cellcolor{brown!25}6 & \cellcolor{red!25}1 & \cellcolor{green!25}2 & \cellcolor{magenta!25}5  \\ \hline
\end{tabular}
\quad
\begin{tabular}{|cccccccc|} \hline
\cellcolor{blue!25}0 & \cellcolor{red!25}1 & \cellcolor{green!25}2 & \cellcolor{yellow!25}3 & \cellcolor{cyan!25}4 & \cellcolor{magenta!25}5 & \cellcolor{brown!25}6 & \cellcolor{violet!25}7  \\
\cellcolor{red!25}1 & \cellcolor{blue!25}0 & \cellcolor{yellow!25}3 & \cellcolor{green!25}2 & \cellcolor{magenta!25}5 & \cellcolor{cyan!25}4 & \cellcolor{violet!25}7 & \cellcolor{brown!25}6  \\
\cellcolor{green!25}2 & \cellcolor{yellow!25}3 & \cellcolor{blue!25}0 & \cellcolor{red!25}1 & \cellcolor{brown!25}6 & \cellcolor{violet!25}7 & \cellcolor{cyan!25}4 & \cellcolor{magenta!25}5  \\
\cellcolor{yellow!25}3 & \cellcolor{green!25}2 & \cellcolor{red!25}1 & \cellcolor{blue!25}0 & \cellcolor{violet!25}7 & \cellcolor{brown!25}6 & \cellcolor{magenta!25}5 & \cellcolor{cyan!25}4  \\
\cellcolor{cyan!25}4 & \cellcolor{magenta!25}5 & \cellcolor{brown!25}6 & \cellcolor{violet!25}7 & \cellcolor{blue!25}0 & \cellcolor{red!25}1 & \cellcolor{green!25}2 & \cellcolor{yellow!25}3  \\
\cellcolor{magenta!25}5 & \cellcolor{cyan!25}4 & \cellcolor{violet!25}7 & \cellcolor{brown!25}6 & \cellcolor{red!25}1 & \cellcolor{blue!25}0 & \cellcolor{yellow!25}3 & \cellcolor{green!25}2  \\
\cellcolor{brown!25}6 & \cellcolor{violet!25}7 & \cellcolor{cyan!25}4 & \cellcolor{magenta!25}5 & \cellcolor{green!25}2 & \cellcolor{yellow!25}3 & \cellcolor{blue!25}0 & \cellcolor{red!25}1  \\
\cellcolor{violet!25}7 & \cellcolor{brown!25}6 & \cellcolor{magenta!25}5 & \cellcolor{cyan!25}4 & \cellcolor{yellow!25}3 & \cellcolor{green!25}2 & \cellcolor{red!25}1 & \cellcolor{blue!25}0  \\ \hline
\end{tabular}
\caption{A set of $\mols{3}{8}$ found by our search.}
\end{figure}

\end{document}